\documentclass{eptcs-modified}
\usepackage{amsmath}
\usepackage{amssymb}
\usepackage{amsthm}
\usepackage{amscd}

\title{How constructive is constructing measures?}
\author{
Arno Pauly
\institute{Computer Laboratory\\ University of Cambridge, United Kingdom
\email{Arno.Pauly@cl.cam.ac.uk}}
\and
Willem L. Fouch\'e
\institute{Department of Decision Sciences\\ School of Economic Sciences, University of South Africa
\email{fouchwl@unisa.ac.za}}
}
\def\titlerunning{Constructing measures}
\def\authorrunning{A. Pauly \& W. Fouch\'e}

\begin{document} \theoremstyle{definition}
\newtheorem{theorem}{Theorem}
\newtheorem{definition}[theorem]{Definition}
\newtheorem{problem}[theorem]{Problem}
\newtheorem{assumption}[theorem]{Assumption}
\newtheorem{corollary}[theorem]{Corollary}
\newtheorem{proposition}[theorem]{Proposition}
\newtheorem{lemma}[theorem]{Lemma}
\newtheorem{observation}[theorem]{Observation}
\newtheorem{fact}[theorem]{Fact}
\newtheorem{question}[theorem]{Question}
\newtheorem{example}[theorem]{Example}
\newcommand{\dom}{\operatorname{dom}}
\newcommand{\Dom}{\operatorname{Dom}}
\newcommand{\codom}{\operatorname{CDom}}
\newcommand{\id}{\textnormal{id}}
\newcommand{\Cantor}{{\{0, 1\}^\mathbb{N}}}
\newcommand{\Baire}{{\mathbb{N}^\mathbb{N}}}
\newcommand{\uint}{{[0,1]}}
\newcommand{\lev}{\textnormal{Lev}}
\newcommand{\hide}[1]{}
\newcommand{\mto}{\rightrightarrows}
\newcommand{\pls}{\textsc{PLS}}
\newcommand{\ppad}{\textsc{PPAD}}
\newcommand{\sierp}{Sierpi\'nski}
\newcommand{\C}{\textrm{C}}
\newcommand{\lpo}{\textrm{LPO}}
\newcommand{\name}[1]{\textsc{#1}}
\newcommand{\me}{\name{P}.~}
\newcommand{\Sierp}{Sierpi\'nski }
\newcommand{\leqW}{\leq_{\textrm{W}}}
\newcommand{\leW}{<_{\textrm{W}}}
\newcommand{\equivW}{\equiv_{\textrm{W}}}
\newcommand{\geqW}{\geq_{\textrm{W}}}
\newcommand{\pipeW}{|_{\textrm{W}}}

\maketitle

\begin{abstract}
Given some set, how hard is it to construct a measure supported by it? We classify some variations of this task in the Weihrauch lattice. Particular attention is paid to Frostman measures on sets with positive Hausdorff dimension. As a side result, the Weihrauch degree of Hausdorff dimension itself is determined.
\end{abstract}

\section{Introduction}
We investigate variations on the problem to construct a measure with a given support. The variations include the available information about the set, whether the support has to be precisely the given set or merely a subset, and whether the measure is required to be non-atomic. A special case of particular importance is the Frostman lemma, which links having certain Hausdorff dimension to admitting a measure with certain properties.

Two of these variations are computable: Given an overt set $A$, one can construct a measure with support precisely $A$. In a very restricted setting, a computable version of the Frostman lemma is available. Apart from these cases, the problems generally are non-computable. Using the framework of Weihrauch reducibility, we can establish the precise degree of non-computability of each case.

The Weihrauch degrees from the framework for the research programme to classify the computational content of mathematical theorems formulated by \name{Brattka} and \name{Gherardi} \cite{brattka3} (also \name{Gherardi} \& \name{Marcone} \cite{gherardi}, \me \cite{paulyincomputabilitynashequilibria}). The core idea is that $S$ is Weihrauch reducible to $T$, if $S$ can be solved using a single invocation of $T$ and otherwise computable means.

Numerous theorems have been classified so far. Some example are the separable Hahn-Banach theorem (\name{Gherardi} \& \name{Marcone} \cite{gherardi}), the Intermediate Value Theorem (\name{Brattka} \& \name{Gherardi} \cite{brattka3}), Nash's theorem for bimatrix games (\me \cite{paulyincomputabilitynashequilibria}), Brouwer's Fixed Point theorem (\name{Brattka}, \name{Le Roux} \& \me \cite{paulybrattka3cie}), the Bolzano-Weierstrass theorem (\name{Brattka}, \name{Gherardi} \& \name{Marcone} \cite{gherardi4}), the Radon-Nikodym derivative (\name{Hoyrup}, \name{Rojas} \& \name{Weihrauch} \cite{hoyrup2b}), the Lebesgue Density Lemma (\name{Brattka}, \name{Gherardi} \& \name{H\"olzl} \cite{hoelzl}), the Goerde-Browder-Kirk fixed point theorem (\name{Neumann} \cite{eike-neumann}) and variants of determinacy of infinite sequential games (\name{Le Roux} and \me \cite{paulyleroux3-arxiv}).

\section{Background}
\subsection{A short introduction to represented spaces}
We briefly present some fundamental concepts on represented spaces following \cite{pauly-synthetic-arxiv}, to which the reader shall also be referred for a more detailed presentation. The concept behind represented spaces essentially goes back to \name{Weihrauch} and \name{Kreitz} \cite{kreitz}, the name may have first been used by \name{Brattka} \cite{brattka13}. A \emph{represented space} is a pair $\mathbf{X} = (X, \delta_X)$ of a set $X$ and a partial surjection $\delta_X : \subseteq \Baire \to X$. A function between represented spaces is a function between the underlying sets. For $f : \mathbf{X} \to \mathbf{Y}$ and $F : \subseteq \Baire \to \Baire$, we call $F$ a realizer of $f$ (notation $F \vdash f$), iff $\delta_Y(F(p)) = f(\delta_X(p))$ for all $p \in \dom(f\delta_X)$, i.e.~if the following diagram commutes:
 $$\begin{CD}
\Baire @>F>> \Baire\\
@VV\delta_\mathbf{X}V @VV\delta_\mathbf{Y}V\\
\mathbf{X} @>f>> \mathbf{Y}
\end{CD}$$
A map between represented spaces is called computable (continuous), iff it has a computable (continuous) realizer. Similarly, we call a point $x \in \mathbf{X}$ computable, iff there is some computable $p \in \Baire$ with $\delta_\mathbf{X}(p) = x$. A priori, the notion of a continuous map between represented spaces and a continuous map between topological spaces are distinct and should not be confused!

Given two represented spaces $\mathbf{X}$, $\mathbf{Y}$ we obtain a third represented space $\mathcal{C}(\mathbf{X}, \mathbf{Y})$ of functions from $X$ to $Y$ by letting $0^n1p$ be a $[\delta_X \to \delta_Y]$-name for $f$, if the $n$-th Turing machine equipped with the oracle $p$ computes a realizer for $f$. As a consequence of the UTM theorem, $\mathcal{C}(-, -)$ is the exponential in the category of continuous maps between represented spaces, and the evaluation map is even computable (as are the other canonic maps, e.g.~currying).

This function space constructor, together with two represented spaces , $\mathbb{N} = (\mathbb{N}, \delta_\mathbb{N})$ and $\mathbb{S} = (\{\bot, \top\}, \delta_\mathbb{S})$, allows us to obtain a model of \name{Escard\'o}'s synthetic topology \cite{escardo}. The representation are given by $\delta_\mathbb{N}(0^n10^\mathbb{N}) = n$, $\delta_\mathbb{S}(0^\mathbb{N}) = \bot$ and $\delta_\mathbb{S}(p) = \top$ for $p \neq 0^\mathbb{N}$. It is straightforward to verify that the computability notion for the represented space $\mathbb{N}$ coincides with classical computability over the natural numbers. The \Sierp space $\mathbb{S}$ in turn allows us to formalize semi-decidability. The computable functions $f : \mathbb{N} \to \mathbb{S}$ are exactly those where $f^{-1}(\{\top\})$ is recursively enumerable (and thus $f^{-1}(\{\bot\})$ co-recursively enumerable).

In general, for any represented space $\mathbf{X}$ we obtain two spaces of subsets of $\mathbf{X}$; the space of open sets $\mathcal{O}(\mathbf{X})$ by identifying $f \in \mathcal{C}(\mathbf{X}, \mathbb{S})$ with $f^{-1}(\{\top\})$, and the space of closed\footnote{There is a very unfortunate confusing nomenclature here. The computable closed sets in our terminology are often called \emph{co-c.e. closed} following \name{Weihrauch}. The sets \name{Weihrauch} calls \emph{computably closed} are the computable closed and overt sets introduced below.} sets $\mathcal{A}(\mathbf{X})$ by identifying $f \in \mathcal{C}(\mathbf{X}, \mathbb{S})$ with $f^{-1}(\{\bot\})$. The properties of the spaces of open and closed sets, namely computability of the usual operations, follow from computability of the functions $\wedge, \vee : \mathbb{S} \times \mathbb{S} \to \mathbb{S}$ and $\bigvee : \mathcal{C}(\mathbb{N}, \mathbb{S}) \to \mathbb{S}$.

One useful consequence of staying within the category when forming the space $\mathcal{O}(\mathbf{X})$ is that we can iterate this to obtain $\mathcal{O}(\mathcal{O}(\mathbf{X}))$, a space appearing in several further constructions. We introduce the space $\mathcal{K}(\mathbf{X})$ of compact sets by identifying a set $A \subseteq \mathbf{X}$ with $\{U \in \mathcal{O}(\mathbf{X}) \mid A \subseteq U\} \in \mathcal{O}(\mathcal{O}(\mathbf{X})$. To ensure well-definedness, we restrict the sets $A$ to saturated sets. As a dual notion, we find the space of overt set $\mathcal{V}(\mathbf{X})$ by identifying $A \subseteq \mathbf{X}$ with $\{U \in \mathcal{O}(\mathbf{X}) \mid A \cap U \neq \emptyset\} \in \mathcal{O}(\mathcal{O}(\mathbf{X}))$. The canonization operation here is the topological closure.

We will the operation $\wedge$ on represented spaces. Given two spaces $\mathbf{X} = (X, \delta_\mathbf{X})$ and $\mathbf{Y} (Y, \delta_\mathbf{Y})$, the underlying set of $\mathbf{X} \wedge \mathbf{Y}$ is $X \cap Y$, and the representation $\delta_\wedge$ is given by $\delta_\wedge(\langle p, q\rangle) = x$ iff $\delta_\mathbf{X}(p) = x \ \wedge \ \delta_\mathbf{Y}(q) = x$. This construction is used to introduce the space $\mathcal{A}(\mathbf{X}) \wedge \mathcal{V}(\mathbf{X})$ of closed and overt subsets.

There always is a canonic computable map $\kappa_\mathbf{X} : \mathbf{X} \to \mathcal{O}(\mathcal{O}(\mathbf{X}))$ defined via $\kappa_\mathbf{X}(x) = \{U \mid x \in U\}$. Using the spaces introduced above, we can read $\kappa_\mathbf{X} : \mathbf{X} \to \mathcal{K}(\mathbf{X})$ as $\kappa_\mathbf{X}(x) = \{x\}\uparrow$ or $\kappa_\mathbf{X} : \mathbf{X} \to \mathcal{V}(\mathbf{X})$ as $\kappa_\mathbf{X}(\{x\}) = \overline{\{x\}}$ instead. The image of $\mathbf{X}$ under $\kappa_\mathbf{X}$ shall be denoted by $\mathbf{X}_\kappa$. The following definition essentially goes back to Schr\"oder \cite{schroder5} and provides an effective counterpart to the definition in \cite{schroder}:

\begin{definition}
A space $\mathbf{X}$ is called \emph{computably admissible}, if $\mathbf{X}$ and $\mathbf{X}_\kappa$ are computably isomorphic.
\end{definition}

Note that $\mathbf{X}_\kappa$ is always computably admissible, i.e.~isomorphic to $(\mathbf{X}_\kappa)_\kappa$. The computably admissible spaces are precisely those that can be regarded as topological spaces, based on the fact that the computable map $f \mapsto f^{-1} : \mathcal{C}(\mathbf{X}, \mathbf{Y}) \to \mathcal{C}(\mathcal{O}(\mathbf{Y}), \mathcal{O}(\mathbf{X}))$ becomes computably invertible iff $\mathbf{Y}$ is computably admissible.

As a special case of represented spaces, we define computable metric spaces following \name{Weihrauch}'s \cite{weihrauchi}. The computable Polish spaces, are derived from complete computable metric spaces by forgetting the details of the metric, and just retaining the representation (or rather, the equivalence class of representations under computable translations).

\begin{definition}
\label{def:cms}
We define a computable metric space with its Cauchy representation such that:
\begin{enumerate}
\item A computable metric space is a tuple $\mathbf{M} = (M, d, (a_n)_{n \in \mathbb{N}})$ such that $(M,d)$ is a metric space and $(a_n)_{n \in\mathbb{N}}$ is a dense sequence in $(M,d)$.
\item The relation
\[ \{(t,u,v,w) \: |\: \nu_{\mathbb{Q}}(t) < d(a_u, a_v) <\nu_{\mathbb{Q}}(w) \} \text{ is recursively enumerable.} \]
\item The Cauchy representation $ \delta_{\mathbf{M}} \: : \: \Baire \rightharpoonup M $ associated with the computable metric space $ \mathbf{M} =  (M, d, A, \alpha) $ is defined by
\[ \delta_{\mathbf{M}}(p) = x \: : \: \Longleftrightarrow  \begin{cases}
      d(a_{p(i)}, a_{p(k)}) \leq 2^{-i} \text{ for } i < k\\
   \text{and } x = \lim\limits_{i\rightarrow \infty}a_{p(i)}
  \end{cases} \]

\end{enumerate}
\end{definition}

As any computable metric space is (effectively) countably based, the following well-known characterization of the open sets in such spaces is useful for us:
\begin{proposition}[{\cite[Proposition 11]{pauly-gregoriades-arxiv}}]
Let $\mathbf{X}$ have an effective countable base $(U_i)_{i \in \mathbb{N}}$ and be computably separable. Then the map $\bigcup : \mathcal{O}(\mathbb{N}) \to \mathcal{O}(\mathbf{X})$ defined via $\bigcup (S) = \bigcup_{i \in S} U_i$ is computable and has a computable multivalued inverse.
\end{proposition}
\subsection{Computable measure theory}
Computable measure theory requires a further special represented space for its development. We can introduce the space $\mathbb{R}_<$ by identifying  a real number $x$ with the set $\{ y \in \mathbb{R} \mid y < x\} \in \mathcal{O}(\mathbb{R})$. Equivalently, using $\{ y \in \mathbb{Q} \mid y < x\} \in \mathcal{O}(\mathbb{Q})$ provides the same result. A third way is to use a monotone growing sequence $(q_n)_{n \in \mathbb{N}} \in \mathcal{C}(\mathbb{N}, \mathbb{Q})$ as a stand-in for $\sup_{n \in \mathbb{N}} q_n \in \mathbb{R}$. We will make use of the following:

\begin{lemma}
\label{lemma:opensum}
$U \mapsto \{y \in \mathbb{R} \mid \exists (x_0, x_1, \ldots) \in U \ \left (\sum_{i \in \mathbb{N}} x_i \leq y \right )\} : \mathcal{O}(\uint^\mathbb{N}) \to \mathcal{O}(\mathbb{R}_<)$ is computable.
\begin{proof}
Using type conversion, we show instead that given $U \in \mathcal{O}(\uint^\mathbb{N})$ and $y \in \mathbb{R}$, it is recognizable if $\exists (x_0, x_1, \ldots) \in U \ \left (\sum_{i \in \mathbb{N}} x_i \leq y \right )$. Given $y$, we can simultaneously try $(q_0, q_1, \ldots, q_n, 0, 0, \ldots) \in U?$ for all rational vectors $(q_0, \ldots, q_n)$ such that $y > \sum_{i=0}^n q_i$. If we find such a vector, then clearly the answer is \textsc{yes}. On the other hand, if any $(x_0, x_1, \ldots) \in U$ with $\left (\sum_{i \in \mathbb{N}} x_i \leq y \right )$ exists, then there must a rational eventually-zero such vector since $U$ is open.
\end{proof}
\end{lemma}

Given some represented space $\mathbf{X}$, we direct our attention to the space $\mathcal{C}(\mathcal{O}(\mathbf{X}), \mathbb{R}_<)$ of continuous functions from the open subsets of $\mathbf{X}$ to $\mathbb{R}_<$. Note that for any $\mu \in \mathcal{C}(\mathcal{O}(\mathbf{X}), \mathbb{R}_<)$ we find that if $U \subseteq V$ for some $U, V \in \mathcal{O}(\mathbf{X})$, then $\mu(U) \leq \mu(V)$. We introduce the space $\mathcal{M}(\mathbf{X})$ as a subspace of $\mathcal{C}(\mathcal{O}(\mathbf{X}), \mathbb{R}_<)$ by: \[\mathcal{M}(\mathbf{X}) := \{\mu \in \mathcal{C}(\mathcal{O}(\mathbf{X}), \mathbb{R}_<) \mid \begin{array}{l} \mu(\emptyset) = 0 \quad \wedge \\ \forall (U_i)_{i \in \mathbb{N}} \in \mathcal{C}(\mathbb{N},\mathcal{O}(\mathbf{X}) \left (\forall i \neq j \in \mathbb{N} \ U_i \cap U_j = \emptyset \right ) \Rightarrow \left ( \mu(\bigcup_{i \in \mathbb{N}} U_i) = \sum_{i \in \mathbb{N}} \mu(U_i) \right )\end{array}\}\]
The space $\mathcal{P}(\mathbf{X})$ of probability measures is obtained in the straightforward way as $\mathcal{P}(\mathbf{X}) := \{\mu \in \mathcal{M}(\mathbf{X}) \mid \mu(\mathbf{X}) = 1\}$.

Given a point $x \in \mathbf{X}$, we can define the point-measure $\pi_x$ by $\pi_x(U) = 1$ iff $x \in U$ and $\pi_x(U) = 0$ otherwise. Then $x \mapsto \pi_x : \mathbf{X} \to \mathcal{P}(\mathbf{X})$ is computable. Also, the usual push-forward operation $(f, \mu) \mapsto f^*\mu : \mathcal{C}(\mathbf{X}, \mathbf{Y}) \times \mathcal{M}(\mathbf{X}) \to \mathcal{M}(\mathbf{Y})$ is computable. However, this works only for the continuous functions, not for any larger class of measurable functions:

\begin{definition}
Let $\mathcal{P}(\mathbf{X})$ be the space of probability measures on $\mathbf{X}$. To define the space $\mathcal{M}^C(\mathbf{X},\mathbf{Y})$ of measurable functions from $\mathbf{X}$ to $\mathbf{Y}$, identify a measurable function $f :\mathbf{X} \to \mathbf{Y}$ with its lifted version $f^* : \mathcal{P}(\mathbf{Y}) \to \mathcal{P}(\mathbf{X})$, and define $\mathcal{M}^C(\mathbf{X},\mathbf{Y})$ as the according subspace of $\mathcal{C}(\mathcal{P}(\mathbf{Y}),\mathcal{P}(\mathbf{X}))$.
\end{definition}

\begin{proposition}
Let $\mathbf{Y}$ be computably admissible. Then $\mathcal{C}(\mathbf{X},\mathbf{Y}) \cong \mathcal{M}^C(\mathbf{X},\mathbf{Y})$.
\begin{proof}
Given a continuous function $f :\mathbf{X}\to\mathbf{Y}$, we can get $f^{-1} :\mathcal{O}(\mathbf{Y}) \to \mathcal{O}(\mathbf{X})$, and compose this with a measure $\nu \in \mathcal{P}(\mathbf{Y})$ to obtain $f^*\nu$. This establishes one direction.

For the other direction, note that $x \mapsto \pi_x : \mathbf{X} \to \mathcal{P}(\mathbf{X})$ is computable, where $\pi_x$ is the point measure at $x$. Now $f^*\pi_x(U) > 0 \Leftrightarrow f(x) \in U$. The left hand side is recognizable by the definition of $\mathcal{P}(\mathbf{Y})$, and admissibility of $\mathbf{Y}$ means that the recognizability of the right hand side for arbitrary $U$ implies continuity of $f$.
\end{proof}
\end{proposition}

As a representation always is a continuous function (by definition of continuity), we see that we can push a measure on Baire space out to the represented space. As shown by \name{Schr\"oder}, in many cases the converse is also true:

\begin{theorem}[\name{Schr\"oder} \cite{schroder2}]
Let $\mathbf{X}$ be a complete computably admissible space. The map $\mu \mapsto \delta_\mathbf{X}^*\mu : \mathcal{P}(\Baire) \to \mathcal{P}(\mathbf{X})$ is computable and computably invertible.
\end{theorem}

In a computable Polish space it is sufficient to know the value a measure takes on the basic open balls, as the following (simple) generalization of a result by \name{Weihrauch} \cite{weihrauche} shows:

\begin{proposition}
Let $\mathbf{X}$ be a computable Polish space with dense sequence $(a_n)_{n \in \mathbb{N}}$. Then the map $\mu \mapsto (\mu(B(a_n,2^{-k})))_{\langle n,k \rangle \in \mathbb{N}} : \mathcal{M}(\mathbf{X}) \to \mathcal{C}(\mathbb{N}, \mathbb{R}_<)$ is computable and computably invertible.
\begin{proof}
That this map is computable is straight-forward element-wise function application. For the inverse direction, note that given some open set $U$ we can effectively approximate it from the inside by finite disjoint unions of basic open balls. Both finite sums and countable suprema are computable on $\mathbb{R}_<$, and this is all that is required.
\end{proof}
\end{proposition}

We will frequently use the preceding proposition, and construct measures simply by providing their values on the basic open balls without further notice.

In the (albeit very restricted) case of probability measures on $\mathbb{N}$ there is further interesting characterization available. Essentially, one may use typical sequences as names for measures:
\begin{theorem}[\cite{paulymeasurement}]
Uniformly in $\varepsilon > 0$ there is a computable and computably invertible function $S_\varepsilon : \subseteq \Baire \to \mathcal{P}(\mathbb{N})$ such that for all $\mu \in \mathcal{P}(\mathbb{N})$ we find that $\widehat{\mu}(S_\varepsilon^{-1}(\mu)) \geq 1 - \varepsilon$, and if $\nu \neq \mu$, then $\widehat{\nu}(S_\varepsilon^{-1}(\mu)) = 0$. Here $\widehat{\mu}$ denotes the induced product measure on $\Baire$.
\end{theorem}

\subsection{Weihrauch reducibility}
\label{subsec:weihrauch}
\begin{definition}[Weihrauch reducibility]
\label{def:weihrauch}
Let $f,g$ be multi-valued functions on represented spaces.
Then $f$ is said to be {\em Weihrauch reducible} to $g$, in symbols $f\leqW g$, if there are computable
functions $K,H:\subseteq\Baire\to\Baire$ such that $K\langle \id, GH \rangle \vdash f$ for all $G \vdash g$.
\end{definition}
The relation $\leqW$ is reflexive and transitive. We use $\equivW$ to denote equivalence regarding $\leqW$,
and by $\leW$ we denote strict reducibility. By $\mathfrak{W}$ we refer to the partially ordered set of equivalence classes. As shown in \cite{paulyreducibilitylattice,brattka2}, $\mathfrak{W}$ is a distributive lattice, and also the usual product operation on multivalued function induces an operation $\times$ on $\mathfrak{W}$. The algebraic structure on $\mathfrak{W}$ has been investigated in further detail in \cite{paulykojiro,paulybrattka4}.

We will make use of an operation $\star$ defined on $\mathfrak{W}$ that captures aspects of function decomposition. Following \cite{gherardi4,paulybrattka3cie}, let $f \star g := \max_{\leqW} \{f_0 \circ g_0 \mid f \equivW f_0 \wedge g \equivW g_0\}$. We understand that the quantification is running over all suitable functions $f_0$, $g_0$ with matching types for the function decomposition. It is not obvious that this maximum always exists, this is shown in \cite{paulybrattka4} using an explicit construction for $f \star g$. Like function composition, $\star$ is associative but generally not commutative.

An important source for examples of Weihrauch degrees relevant in order to classify theorems are the closed choice principles studied in e.g.~\cite{brattka3,paulybrattka}:
\begin{definition}
Given a represented space $\mathbf{X}$, the associated closed choice principle $\C_\mathbf{X}$ is the partial multivalued function $\C_\mathbf{X} : \subseteq \mathcal{A}(\mathbf{X}) \mto \mathbf{X}$ mapping a non-empty closed set to an arbitrary point in it.
\end{definition}

For any uncountable compact metric space $\mathbf{X}$ we find that $\C_\mathbf{X} \equivW \C_\uint$. For well-behaved spaces, using closed choice iteratively does not increase its power, in particular $\C_\mathbb{N} \star \C_\mathbb{N} \equivW \C_\mathbb{N}$ and $\C_\uint \star \C_\uint \equivW \C_\uint$. Likewise, it was shown that $\C_{\mathbb{R}^n} \equivW \C_{\mathbb{R}^n} \star \C_{\mathbb{R}^n} \equivW \C_\mathbb{N} \times \C_\uint \equivW \C_\mathbb{N} \star \C_\uint \equivW \C_\uint \star \C_\mathbb{N}$ for any $n > 0$. Closed choice for $\uint$ and $\Cantor$ is incomparable. The degree $\C_\uint$ is closely linked to $\textrm{WKL}$ in reverse mathematics, while $\C_\mathbb{N}$ is Weihrauch-complete for functions computable with finitely many mindchanges.

Further variations of closed principle providing a fruitful area of study are obtained by restriction to certain subclasses of the closed sets. In \cite{paulybrattka3,paulybrattka3cie} choice for connected closed subsets of $\uint^k$ was studied (and related to Brouwer's Fixed Point theorem). Convex and finite sets were compared in \cite{paulyleroux-cie,paulyleroux-arxiv}. Most related to the present investigation, choice for sets of positive Lebesgue measure was studied in \cite{paulybrattka2,hoelzl}. This yields a Weihrauch degree $\textrm{PC}_{\uint}$ with $\textrm{PC}_\uint \leW \C_\uint$ and $\textrm{PC}_\uint \pipeW \C_\mathbb{N}$. Once more, replacing $\uint$ with another uncountable compact metric space $\mathbf{X}$ does not change the Weihrauch degree.

Another typical degree is obtained from the limit operator $\lim : \subseteq \Baire \to \Baire$ defined via $\lim(p)(n) = \lim_{i \to \infty} p(\langle n, i\rangle)$. This degree was studied  by \name{von Stein} \cite{stein}, \name{Mylatz} \cite{mylatz} and \name{Brattka} \cite{brattka11,brattka}, with the latter noting in \cite{brattka} that it is closely connected to the Borel hierarchy. \name{Hoyrup}, \name{Rojas} and \name{Weihrauch} have shown that $\lim$ is equivalent the Radon-Nikodym derivative in \cite{hoyrup2b}. It also appears in the context of model of hypercomputation as shown by \name{Ziegler} \cite{ziegler2,ziegler7}, and captures precisely the additional computational power certain solutions to general relativity could provide beyond computability \cite{hogarth}. It is related to the examples above via $\C_\mathbb{N} \times \C_\uint \leW \lim \equivW \lim \times \lim \leW \lim \star \lim$.

Important further representatives of the degree of $\lim$ are found in the following:
\begin{theorem}[\name{v Stein} \cite{stein}]
\label{theo:vstein}
$\left (\id : \mathcal{A}(\mathbb{N}) \to \mathcal{V}(\mathbb{N})\right ) \equivW \left (\id : \mathcal{V}(\mathbb{N}) \to \mathcal{A}(\mathbb{N})\right ) \equivW \left ( \id : \mathbb{R}_< \to \mathbb{R} \right) \equivW \lim$
\end{theorem}

This can be generalized further:
\begin{proposition}
Let $\mathbf{X}$ be a computable metric space admitting a computable sequence $a \in \mathcal{C}(\mathbb{N},\mathbf{X})$ together with a computable sequence $r \in \Baire$ such that $\forall n, m \in \mathbb{N} \ \left (n \neq m \Rightarrow d(a_n, a_m) > 2^{-r_n} \right )$. Then: \[\lim \leqW \left (\id : \mathcal{A}(\mathbf{X}) \to \mathcal{V}(\mathbf{X}) \right )\]
\begin{proof}
We show the reduction $\left (\id : \mathcal{A}(\mathbb{N}) \to \mathcal{V}(\mathbb{N})\right ) \leqW \left (\id : \mathcal{A}(\mathbf{X}) \to \mathcal{V}(\mathbf{X}) \right )$ instead. Given some closed set $A \in \mathcal{A}(\mathbb{N})$, we construct a closed set $B \in \mathcal{A}(\mathbf{X})$ as follows: When $n \in \mathbb{N}$ is removed from $A$, we remove $B(a_n, 2^{-r_n})$ from $B$. By choice of the sequence, we then find that $n \in A \Leftrightarrow B(x_n,2^{-r_n}) \cap B \neq \emptyset$, hence knowing $B$ as an overt set implies knowing $A$ as an overt set.
\end{proof}
\end{proposition}

Whether any infinite computable metric space admits a computable sequence as above seems to be an open problem. Certainly such sequences exist without the computability requirement, thus we obtain:

\begin{corollary}
Let $\mathbf{X}$ be an infinite computable metric space. Then $\lim \leqW \left (\id : \mathcal{A}(\mathbf{X}) \to \mathcal{V}(\mathbf{X}) \right )$ relative to some oracle.
\end{corollary}

\begin{proposition}
Let $\mathbf{X}$ be a locally compact computable metric space. Then:
\[\left (\id : \mathcal{A}(\mathbf{X}) \to \mathcal{V}(\mathbf{X}) \right ) \leqW \lim\]
\begin{proof}
We may assume a basis of basic open balls with compact closures in this case. Then, given $A \in \mathcal{A}(\mathbf{X})$, we can enumerate all basic open balls $B_n$ such that $\overline{B_n} \cap A = \emptyset$. Using $\left (\id : \mathcal{A}(\mathbb{N}) \to \mathcal{V}(\mathbb{N})\right ) \equivW \lim$, we can transform such an enumeration to its characteristic function $\chi \in \Cantor$ (i.e.~$\chi_n = 1 \Leftrightarrow \overline{B_n} \cap A = \emptyset$.

Now if for some open $U \in \mathcal{O}(\mathbf{X})$ we have $U \cap A \neq \emptyset$, then there is some $B_n$ with $\overline{B_n} \subseteq U$ and $\overline{B_n} \cap A \neq \emptyset$. Given $\chi$, we can effectively search for such a candidate, thus, $A \in \mathcal{V}(\mathbf{X})$ is computable from $\chi$.
\end{proof}
\end{proposition}

\begin{corollary}
Let $\mathbf{X}$ be an infinite locally compact computable metric space. Then relative to some oracle:
\[\left (\id : \mathcal{A}(\mathbf{X}) \to \mathcal{V}(\mathbf{X}) \right ) \equivW \lim\]
\end{corollary}

\begin{proposition}
Let $\mathbf{X}$ admit a computable sequence $(a_n)_{n \in \mathbb{N}}$ such that $\forall n a_n \notin \overline{\{a_i \mid i \neq n\}}$. Then:
\[\lim \leqW \left (\id : \mathcal{V}(\mathbf{X}) \to \mathcal{A}(\mathbf{X}) \right )\]
\begin{proof}
As the (closure of the) image of an overt set under a continuous function is overt, the map $A \mapsto \overline\{a_i \mid i \in A\} : \mathcal{V}(\mathbb{N}) \to \mathcal{V}(\mathbf{X})$ is computable. By assumption on $(a_n)_{n \in \mathbb{N}}$, we now find that $n \notin A \Leftrightarrow a_n \notin \overline{\{a_i \mid i \in A\}}$. Thus, we have a reduction: \[\left (\id : \mathcal{V}(\mathbb{N}) \to \mathcal{A}(\mathbb{N}) \right ) \leqW \left (\id : \mathcal{V}(\mathbf{X}) \to \mathcal{A}(\mathbf{X}) \right )\]
By \name{v Stein}'s result (Theorem \ref{theo:vstein}), this is equivalent to our claim.
\end{proof}
\end{proposition}

\begin{proposition}
Let $\mathbf{X}$ be computably countably based. Then:
\[\left (\id : \mathcal{V}(\mathbf{X}) \to \mathcal{A}(\mathbf{X}) \right ) \leqW \lim\]
\begin{proof}
Let $(U_n)_{n \in\mathbb{N}}$ be a computable countable basis. Given $A \in \mathcal{V}(\mathbf{X})$, we may compute $L := \{n \in \mathbb{N} \mid U_n \cap A \neq \emptyset\} \in \mathcal{O}(\mathbb{N}) \equiv \mathcal{V}(\mathbb{N})$. Using $\left (\id : \mathcal{V}(\mathbb{N}) \to \mathcal{A}(\mathbb{N}) \right )$, we find $L \in \mathcal{A}(\mathbb{N})$, or equivalently, $L^C \in \mathcal{O}(\mathbb{N})$. Now we may compute $\left (\bigcup_{i \in L^C} U_i \right ) \in \mathcal{O}(\mathbf{X})$, and will find this to be equivalent to $A \in \mathcal{A}(\mathbf{X})$.
\end{proof}
\end{proposition}

\begin{corollary}
$\left (\id : \mathcal{A}(\uint) \to \mathcal{V}(\uint) \right ) \equivW \left (\id : \mathcal{V}(\uint) \to \mathcal{A}(\uint) \right ) \equivW \lim$
\end{corollary}

\subsection{Hausdorff dimension and the Frostman lemma}
\label{subsec:introhausdorff}
The Frostman lemma essentially states that having positive Hausdorff dimension is equivalent to admitting a measure that is far from being atomic -- and \emph{being far from atomic} is given a quantitative interpretation and exactly tied to the Hausdorff dimension. We introduce the Hausdorff dimension only as a property of closed subsets of $\uint$ here:

\begin{definition}
\label{def:hausdorff}
Given some $A \in \mathcal{A}(\uint)$, we define its Hausdorff dimension $\dim_\mathcal{H}(A)$ as:
\[\begin{array}{ccl} \dim_\mathcal{H}(A) & := & \inf \{d \geq 0 \mid \inf \{\sum_{i \in \mathbb{N}} r_i^d \mid \exists (x_i)_{i \in \mathbb{N}} \ A \subseteq \bigcup_{i \in \mathbb{N}} B(x_i,r_i) \ \} = 0\} \\ & = & \sup \{d \geq 0 \mid \lim_{\delta \to 0} \inf \{\sum_{i \in \mathbb{N}} r_i^d \mid \exists (x_i)_{i \in \mathbb{N}} \ A \subseteq \bigcup_{i \in \mathbb{N}} B(x_i,r_i) \wedge \forall i \ r_i < \delta\ \} = \infty\}\end{array}\]
\end{definition}

We find $\dim_\mathcal{H}(A) \in \uint$ for all $A \in \mathcal{A}(\uint)$, $\dim_\mathcal{H}(A) = 0$ for any countable $A$ and $\dim_\mathcal{H}(A) = 1$ whenever $\lambda(A) > 0$ where $\lambda$ denotes the Lebesgue measure. A way to construction sets of given Hausdorff dimension is provided in \cite{fouche}: Given a sequence $(d_i)_{i \in \mathbb{N}}$ of non-negative reals, we define a family $([a_w,b_w])_{w \in \{0,1\}^*}$ of intervals indexed by $\{0,1\}^*$ by $[a_\varepsilon,b_\varepsilon] = \uint$, $a_{w0} = a_w$, $b_{w0} = a_w + d_{|w|}^{-1}(b_w - a_w)$, $a_{w1} = a_w +
(1 - d_{|w|}^{-1})(b_w - a_w)$, $b_{w1} = b_w$. Then define $C := \bigcap_{n \in \mathbb{N}} \bigcup_{w \in \{0,1\}^n} [a_w, b_w]$, and obtain $\dim_\mathcal{H}(C) = \liminf_{n \in \mathbb{N}} \frac{\ln 2}{\ln d_n}$.

While having positive Hausdorff dimension (and in particular having Hausdorff dimension 1) is an indicator of a set being large, it is a rather weak indicator as the following shows:

\begin{theorem}[\name{Shmerkin} \cite{shmerkin}]
\label{theo:shmerkin}
There is a computable surjection $\vartheta : \uint \to \uint$ such that $\forall x \in \uint \ \dim_\mathcal{H}(\vartheta^{-1}(\{x\})) = 1$.
\begin{proof}
Using the construction above, one can find a subset $C \subseteq \uint$ with $\dim_\mathcal{H}(C) = 1$ such that $C$ is computably isomorphic to $\Cantor$. Let $\gamma : C \to \Cantor$ be the canonic computable isomorphism, and $\delta_2 : \Cantor \to \uint$ the usual binary representation. Consider the computable function $\theta : \Cantor \to \Cantor$ defined via $\theta(p)(n) = p(n^2)$. We obtain $\vartheta$ by extrapolating $\delta_2 \circ \theta \circ \gamma$ in a computable way to $\uint$.

The claim for $\vartheta$ now follows from $\dim_\mathcal{H}( \left (\theta \circ \gamma \right )^{-1}(\{p\})) = 1$ for any $p \in \Cantor$. This can be shown using the mass distribution principle (e.g.~\cite{falconer}). For this, we define a measure $\mu_p$ on $C$ using the intervals $[a_w,b_w]$ occurring in the construction of $C$. Start with $\mu_p(\uint) = 1$. If $|w| = k^2$ for some $k \in \mathbb{N}$, then $\mu_p([a_{w}, b_{w}]) = \mu_p([a_{w_{<|w|}},b_{w_{<|w|}}])$ if $w \prec p$ and $\mu_p([a_{w}, b_{w}]) = 0$ otherwise. If $|w| + 1$ is not a perfect square, then $\mu_p([a_{w0},b_{w0}]) = \mu_p([a_{w1},b_{w1}]) = \frac{1}{2}\mu_p([a_{w},b_{w}])$. This yields indeed a measure, and we find $\mu_p( \left (\theta \circ \gamma \right )^{-1}(\{p\})) = 1$. To invoke the mass distribution principle, we further need that $\lim_{r \to 0} \frac{\log \mu_p(B(x,r))}{\log r} = 1$ for all $x \in \operatorname{supp}(\mu_p)$.
\end{proof}
\end{theorem}

We can now introduce Frostman measures, and state the Frostman lemma:

\begin{definition}
An $s$-Frostman measure (on $\uint$) is a non-zero Radon measure $\mu$ such that for any $x, r \in \uint$ we find $\mu(B(x,r)) \leq r^s$.
\end{definition}

\begin{lemma}[\name{Frostman} \cite{frostman}]
\[\dim_\mathcal{H}(A) = \sup_{s \in \uint} \{\exists \mu \mid \operatorname{supp}(\mu) \subseteq A \wedge \mu \textnormal{ is an $s$-Frostman measure}\}\]
\end{lemma}

We will in particular investigate the Weihrauch degree of the following maps:
\begin{definition}
\label{def:frostmanmaps}
Let $\operatorname{Frost} : \subseteq \mathcal{A}(\uint) \times \uint \mto \mathcal{M}(\uint)$ be defined via $\mu \in \operatorname{Frost}(A,s)$ iff $\mu$ is an $s$-Frostman measure with $\operatorname{supp}(\mu) \subseteq A$. Let $\operatorname{StrictFrost} : \subseteq \mathcal{A}(\uint) \times \uint \mto \mathcal{M}(\uint)$ be defined via $\mu \in \operatorname{Frost}(A,s)$ iff $\mu$ is an $s$-Frostman measure with $\operatorname{supp}(\mu) = A$.
\end{definition}

\section{Measures and support}
\label{sec:genericmeasures}
We shall begin by investigating how a measure and its support are related. We show that the support is fundamentally an overt set, rather than a closed set; and that both obtaining the support of a given measure as a closed set, and constructing a measure with support as a given closed set are equivalent to the $\lim$-operator. The measures constructed here will generally fail to be non-atomic.

\begin{theorem}
\label{theo:suppovert}
$\operatorname{supp} : \mathcal{M}(\mathbf{X}) \to \mathcal{V}(\mathbf{X})$ is computable. If $\mathbf{X}$ is a computable metric space, it has a computable multivalued inverse.
\begin{proof}
Note that an open set $U$ intersects $\operatorname{supp}(\mu)$ iff $\mu(U) > 0$. It is easy to see that that $\mathalpha{>0} : \mathbb{R}_< \to \mathbb{S}$ is computable. Taking into consideration the definitions of $\mathcal{M}$ and $\mathcal{V}$, we see that $f \mapsto (U \mapsto \mathalpha{>0}(f(U)))$ is a computable realizer of $\operatorname{supp}$.

For the reverse we shall construct a probability measure $\mu$ given a non-empty overt set $A \in\mathcal{V}(\mathbf{X})$ such that $\operatorname{supp}(\mu) = A$. Let $(x_i)_{i \in \mathbb{N}}$ be a computable dense sequence in $\mathbf{X}$. We begin by associating numbers $c_{i,k} \in \mathbb{R}_<$ to the basic open balls $B(x_i,k 2^{-k})$. First, we test for any $B(x_i,2^{-1})$ if $B(x_i,2^{-1}) \cap A \neq \emptyset$, producing some infinite sequence $B(x_{l_1},2^{-1}), B(x_{l_2},2^{-1}), \ldots$ (with repetitions) of all those balls intersecting the set. Then we set $c_{i,1} = \sum_{\{j \mid l_j = i\}} 2^{-j+1}$.

In the next round, we test for all $B(x_i,2^{-2})$ if they intersect $A$, and again obtain an infinite sequence $B(x_{l'_1},2^{-2}), B(x_{l'_2},2^{-2}), \ldots$ of balls doing so, potentially with repetitions. For any $i \in\mathbb{N}$, let $\overline{l}_i$ be the least $j$ such that $B(x_{l'_j},2^{-2}) \subseteq B(x_i,2^{-1})$, provided that this exists, and $i$ otherwise. Now we set $c_{i,2} = \sum_{\{j \mid \overline{l}_j = i\}} c_{i,1} + \sum_{\{j \mid l'_j = i\}} 2^{-j+2}$. We proceed to construct the remaining numbers in this pattern. Note that we find $\sum_{i \in \mathbb{N}} c_{i,k} = \frac{1}{2} + \ldots + \frac{1}{2^k}$.

Now our measure is defined by $\mu(B) = \sup_{k \in \mathbb{N}} \sum_{\{i \mid B(x_i,2^{-k}) \subseteq B\}} c_{i,k}$ for any open ball $B$, and then extended in the usual manner.
\end{proof}
\end{theorem}

\begin{corollary}
Consider $\operatorname{supp}^\mathcal{A}_\mathbf{X} : \mathcal{M}(\mathbf{X}) \to \mathcal{A}(\mathbf{X})$ and $\left(\operatorname{supp}^\mathcal{A}_\mathbf{X}\right)^{-1} : \mathcal{A}(\mathbf{X}) \mto \mathcal{M}(\mathbf{X})$. If $\mathbf{X}$ is an infinite locally compact computable Polish space, we find $\operatorname{supp}^\mathcal{A}_\mathbf{X} \equivW \left(\operatorname{supp}^\mathcal{A}_\mathbf{X}\right)^{-1} \equivW \lim$ relative to the Halting problem. Also, $\operatorname{supp}^\mathcal{A}_\uint \equivW \left(\operatorname{supp}^\mathcal{A}_\uint\right)^{-1} \equivW \lim$
\end{corollary}

If we just demand that the (non-zero) measure to be constructed is supported by the given (non-empty) closed set, the resulting operator $\operatorname{ConstructMeasure}_\mathbf{X} :\subseteq \mathcal{A}(\mathbf{X}) \mto \mathcal{M}(\mathbf{X})$ is strictly simpler for many spaces:
\begin{theorem}
Let $\mathbf{X}$ be a computable Polish space. Then $\operatorname{ConstructMeasure}_\mathbf{X} \equivW \C_\mathbf{X}$.
\begin{proof}
For $\operatorname{ConstructMeasure}_\mathbf{X} \leqW \C_\mathbf{X}$, use $\C_\mathbf{X}$ to pick a point $x$ in $A$, and then compute the point measure $\mu_x$, which is non-zero and satisfies $\operatorname{supp}(\mu_x) = \{x\} \subseteq A$.

For the other direction, we need to show that given a non-zero measure $\mu$ supported by $A$ we can compute some point $x \in A$. We do this by searching for a basic open ball $B(x_1,2^{-1})$ with $\mu(B(x_1,2^{-1})) > 0$, which we will find eventually. Then we search for some $x_2 \in B(x_1,2^{-1})$ such that $\mu(B(x_2,2^{-2})) > 0$, which also will eventually be detected. We continue to produce a fast Cauchy sequence, of which we can compute the limit $x$. Now $x \in A$ is easy to see.
\end{proof}
\end{theorem}

\begin{corollary}
$\left (\operatorname{ConstructMeasure}_\uint :\subseteq \mathcal{A}(\uint) \mto \mathcal{M}(\uint) \right ) \leW \left ( \operatorname{supp}^{-1}_\uint : \mathcal{A}(\uint) \mto \mathcal{M}(\uint) \right )$
\end{corollary}

\section{Non-atomic measures}
The picture painted in Section \ref{sec:genericmeasures} of the constructivity (or lack thereof) of constructing measures crucially depends on the option of resulting measures having atoms, i.e.~single points carrying positive measure. In the present section we first introduce the notion of flows on infinite trees as a technical tool (which could be of some interest in its own right). We then proceed to investigate the role of overtness, which drastically differs from the results above. Considering some Weihrauch degrees related to Hausdorff dimension then leads up to the Frostman lemma.
\subsection{Non-atomic measures on $\uint$ and flows}
\label{subsec:flows}
We consider assignments of non-negative real numbers to the edges of a full infinite binary tree, i.e.~the space $\left (\mathbb{R}_0^+\right)^{\{0,1\}^*}$. Such an assignment $f$ is called a \emph{flow}, if $\forall v \in \{0,1\}^* \ f(v) = f(v0) + f(v1)$. We extend $\leq$ to $\left (\mathbb{R}_0^+\right)^{\{0,1\}^*}$ in a point-wise fashion.
\begin{definition}
The multivalued map $\operatorname{MaxFlow} : \left (\mathbb{R}_0^+\right)^{\{0,1\}^*} \mto \left (\mathbb{R}_0^+\right)^{\{0,1\}^*}$ is defined via $g \in \operatorname{MaxFlow}(f)$ iff $g \leq f$, $g$ is a flow and $g(\epsilon)$ is maximal under these conditions.
\end{definition}

\begin{theorem}
$\operatorname{MaxFlow} \equivW \lim$
\end{theorem}

\begin{proof}
To show $\lim \leqW \operatorname{MaxFlow}$, we prove $\id_{\mathbb{R}_>, \mathbb{R}}|_{\{x \in \mathbb{R}_> \mid x \geq 0\}} \leqW \operatorname{MaxFlow}$ instead. Let the input to $\id_{\mathbb{R}_>, \mathbb{R}}|_{\{x \in \mathbb{R}_> \mid x \geq 0\}}$ be the decreasing sequence $(q_i)_{i \in \mathbb{N}}$ of non-negative rationals. Now we construct an assignment via $f(1^n) = q_n$ and $f(v) = 0$ for $v \neq 1^{|v|}$. Now if $g \in \operatorname{MaxFlow}(f)$, then $g(\varepsilon) = \inf_{i \in \mathbb{N}} q_i = \lim_{i \to \infty} q_i$.

For the other direction, we inductively define a decreasing sequence of real numbers for each edge of the tree by setting $a_0^v = f(v)$ and $a_{n+1}^v = \min \{a_n^v, a_n^{v0} + a_n^{v_1}\}$. Then $\left(v \mapsto \lim_{i \to \infty} a_i^v \right ) \in \operatorname{MaxFlow}(f)$, so using $\lim$ countably times in parallel suffices to find a valid max flow. As stated above, this is equivalent to using $\lim$ just once.
\end{proof}

\begin{definition}
The multivalued map $\operatorname{NonZeroFlow} :\subseteq \left (\mathbb{R}_0^+\right)^{\{0,1\}^*} \mto \left (\mathbb{R}_0^+\right)^{\{0,1\}^*}$ is defined via $g \in \operatorname{NonZeroFlow}(f)$ iff $g \leq f$, $g$ is a flow and $g(\epsilon) > 0$.
\end{definition}

\begin{theorem}
$\operatorname{NonZeroFlow} \equivW \C_\mathbb{N} \times \C_\Cantor$.
\begin{proof}
To see that $\operatorname{NonZeroFlow} \leqW \C_\mathbb{N} \times \C_\Cantor$ we use a non-deterministic algorithm following \cite{paulybrattka}. We guess some number $k \in \mathbb{N}$ together with an assignment $g : \{0,1\}^* \to [0;2^{-k}]$ where $g(\varepsilon) = 2^{-k}$ (note that the latter is an element of a computably compact computable metric space). If $g$ is not a flow with $g \leq f$ where $f$ is the input to $\operatorname{NonZeroFlow}$, we will detect this eventually.

For the other direction, we may prove $\textrm{WKL} \times \operatorname{UpperBound} \leqW \operatorname{NonZeroFlow}$ instead. Thus, we are given an infinite binary tree $T$ and a monotone and bounded sequence of natural numbers $(n_i)_{i \in \mathbb{N}}$. Let $\lambda_n^T := |\{v \in T \mid |v| = n\}| - 1$. We define an assignment $f : \{0,1\}^* \to \uint$ via $f(v) = 2^{-n_{|v|}-\lambda_{|v|}^T}$ if $v \in T$ and $f(v) = 0$ otherwise. Any non-zero flow $g$ smaller than $f$ then computes both an infinite path through $T$ (just go down some path carrying positive flow) and an upper bound for $(n_i)_{i \in \mathbb{N}}$ (in form of $N$ s.t.~$g(\varepsilon) > 2^{-N}$).
\end{proof}
\end{theorem}

Say that a flow $f$ is \emph{concentrated}, iff for all $v \in \{0,1\}^*$ either $f(v) \geq 2^{-2|v|-1}$ or $f(v) = 0$. Let the multivalued map $\operatorname{ConcentrateFlow}$ map a non-zero flow $f$ to any non-zero concentrated flow $g$ with $f(\varepsilon)g \leq f$.

\begin{proposition}
\label{prop:concentrateflow}
$\operatorname{ConcentrateFlow}$ is computable.
\begin{proof}
We may normalize the input flow to $f(\varepsilon) = 1$. We define $g$ iteratively, starting with $g(\varepsilon) = \frac{1}{2}$. We want to ensure the invariant $g(v) > 0 \Rightarrow \left (g(v) + 2^{-2|v|-1} \leq f(v) \right ) \wedge \left (2^{-2|v|-1} \leq g(v) \right )$ throughout the process -- any flow constructed this way clearly is a valid answer. The invariant holds at the initial step.

For the continuation step, we can computably select a true case among $f(v0) \geq 3*2^{-2|v|-3}$ or $f(v1) \geq 3*2^{-2|v|-3}$ or $f(v0) \geq 2^{-2|v|-2} \wedge f(v1) \geq 2^{-2|v|-2}$. In the first case, set $g(v0) = g(v)$ and $g(v1) = 0$. In the second case, $g(v0) = 0$ and $g(v1) = g(v)$. In the third case, $g(0) = \min \{f(v0) - 2^{-2|v|-3}, g(v) - 2^{-2|v|-3}\}$ and $g(v1) =g(v) - g(v0)$.
\end{proof}
\end{proposition}
The reason for interests in flows is a connection between flows and (non-atomic) measures on $\uint$. Given $v \in \{0,1\}^*$, let $\mathfrak{D}_v$ be the associated closed dyadic interval. Given some flow $v$, we can define a measure $\mu$ by setting $\mu(\mathfrak{D}_v^\circ) = f(v)$ and then extending it in the usual way. On the other hand, if $\mu$ is a non-atomic measure on $\uint$ with known value $\mu(\uint)$, then $f(v) := \mu(\mathfrak{D}_v^\circ)$ defines a flow.
\subsection{Overtness and non-atomic measures}
It is clear that isolated points cannot be part of the support of an isolated measure. Thus, when constructing measures supported by given sets, we either need to consider only perfect sets, or be satisfied if the support is included in the set, rather than demanding equality. In the latter situation, overtness becomes useless as demonstrated next:

Given some closed set $A$, let $A^*$ be the largest perfect set contained in $A$. Alternatively, let $A^*$ be the set resulting from $A$ after removing the isolated points $\alpha$ times for some countable ordinal $\alpha$, after which no isolated points remain.

\begin{definition}
Define $\textrm{PerfectCore} : \mathcal{A}(\mathbb{R}) \mto \mathcal{A}(\mathbb{R}) \wedge \mathcal{V}(\mathbb{R})$ by $B \in \textrm{PerfectCore}(A)$ if $B \supseteq A$ and $B^* = A^*$.
\end{definition}

\begin{proposition}
\label{prop:perfectcore}
$\textrm{PerfectCore}$ is computable.
\begin{proof}
We assume that we have an enumeration of all basic open intervals whose closure disjoint is with the input $A$, and that we have to decide for each basic open interval whether it intersects the constructed output $B$. Basic open intervals are assumed to be finite, hence have compact closure. We will decide on these answers in some order, and in particular deal with smaller intervals only after those containing it. In the following, let $I$ always the interval we currently need to decide on. We further maintain a set $X$ of rational numbers, which at any stage of the construction will be finite, and initially is empty.
\begin{enumerate}
\item If $I \cap X \neq \emptyset$, then $I \cap B \neq \emptyset$.
\item If we have already learned that $\overline{I} \cap A = \emptyset$, and additionally $I \cap X = \emptyset$, then we decide that $I \cap B = \emptyset$.
\item If $I \cap X = \emptyset$, and it is consistent with the information we have read about $A$ so far that $I \cap A \neq \emptyset$, then we decide that $I \cap B \neq \emptyset$ and keep monitoring $I$ in the following.
\item If $I$ is an interval we monitor due to $(3.)$, and we do learn that $\overline{I} \cap A = \emptyset$, then at this particular time there is still some basic open interval $L \subseteq I$ which is disjoint with any interval $\overline{J}$ for which we already had decided that $\overline{J} \cap B = \emptyset$, due to compactness. We can effectively find such an $L$. Let $x$ be the midpoint of $L$ and add it to the set $X$. Let $L_1$ and $L_2$ be the remaining open halves of $L$. We also decide that $L_1 \cap B = \emptyset$ and $L_2 \cap \emptyset$.
\end{enumerate}
These rules ensure that the decisions made for the individual intervals are actually consistent, and thus the construction actually yields some closed and overt set $B$, which incidentally is of the form $B = A \cup X$ with $X \cap A = \emptyset$. The removal of $L_1$ and $L_2$ in Item $(4.)$ ensures that any $x \in X$ is isolated in $B$, hence $B \in \textrm{PerfectCore}(A)$.
\end{proof}
\end{proposition}

However, in some situations overtness also can be obtained \emph{for free} in the support. We shall call a measure $\nu$ on a suitable metric space $C$-\emph{concentrated}, if for any dyadic open ball $B$ either $\nu(B) > Cr^2$ or $\nu(B) = 0$ for a constant $C > 0$. Define the map $\operatorname{Concentrate} :\subseteq \mathcal{M}(\mathbf{X}) \mto \mathcal{M}(\mathbf{X}) \times \mathbb{N}$ by $\mu \in \dom(\operatorname{Concentrate})$ iff $\mu$ is non-atomic and non-zero. Then let $(\nu, k) \in \operatorname{Concentrate}(\mu)$ if $\nu \leq \mu$, $\nu$ is non-atomic, non-zero and $2^{-k}$-concentrated. Let $\operatorname{ConcentratedSupport} : \subseteq \mathcal{M}(\mathbf{X}) \times \mathbb{N} \to \mathcal{A}(\mathbf{X}) \wedge \mathcal{V}(\mathbf{X})$ be defined via $(\nu,k) \in \dom(\operatorname{ConcentratedSupport})$ iff $\nu$ is $2^{-k}$-concentrated, and $\operatorname{ConcentratedSupport}(\nu,k) = \operatorname{supp}(\nu)$.

\begin{proposition}
\label{prop:concentrate}
Set $\mathbf{X} := \uint$. Then $\operatorname{Concentrate}$ is computable.
\begin{proof}
Let $\mu$ be the input. We search for some $k \in \mathbb{N}$, such that $\mu(\uint) \geq 2^{-k}$. Now we can construct a flow $f$ from $\mu$ such that $f(\varepsilon) = 2^{-k}$ and $f(v) \leq \mu(\mathfrak{D}_v^\circ)$ for any $v \in \{0,1\}^*$. From this flow we obtain a concentrated flow by Proposition \ref{prop:concentrateflow}, which then yields the desired concentrated measure.
\end{proof}
\end{proposition}

\begin{proposition}
\label{prop:concentratedsupport}
Let $\mathbf{X}$ be a computably compact computable metric space. Then $\operatorname{ConcentratedSupport}$ is computable.
\begin{proof}
The $\mathcal{V}$-name is computed as in Theorem \ref{theo:suppovert}. As we work in computably compact space, and are dealing with non-atomic measures, we can actually compute the measure of open balls as real numbers. Then to find the $\mathcal{A}$-name, note that given a concentrated measure $\nu$ for any dyadic open ball $B$ we have that $\nu(B) < Cr^2 \Leftrightarrow B \cap \operatorname{supp}(\nu)$. The left hand side is recognizable, and then the right hand side produces the desired result.
\end{proof}
\end{proposition}

\begin{corollary}
Given a non-atomic probability measure $\mu$ on $\uint$, we can compute a set $A \in \mathcal{V}(\uint) \wedge \mathcal{A}(\uint)$ such that $A \subseteq \operatorname{supp}(\mu)$ and $\mu(A) \geq \frac{1}{2}$.
\end{corollary}

\subsection{Some Weihrauch degrees related to Hausdorff dimension}
As mentioned in Subsection \ref{subsec:weihrauch}, it is often possible to \emph{calibrate} closed choice principles to the strength of a particular theorem by restricted them to a (related) class of closed sets. As such, it seems reasonable to investigate closed choice for subsets of $\uint$ with positive Hausdorff dimension (or even Hausdorff dimension $1$). A straight-forward yet important consequence of \name{Shmerkin}'s construction (Theorem \ref{theo:shmerkin}) is that this does not yield any new Weihrauch degrees:

\begin{definition}
Let $\textrm{HC}_\uint$ be the restriction of $\C_\uint$ to $\{A \in \mathcal{A}(\uint) \mid \dim_\mathcal{H}(A) > 0\}$. Let $\textrm{H}_1\textrm{C}_\uint$ be the restriction of $\C_\uint$ to $\{A \in \mathcal{A}(\uint) \mid \dim_\mathcal{H}(A) = 1\}$.
\end{definition}

\begin{corollary}
$\textrm{HC}_\uint \equivW \textrm{H}_1\textrm{C}_\uint \equivW \C_\uint$.
\end{corollary}

Before continuing with our investigation of measures, we shall classify the Hausdorff dimension itself. Note that our result does make use of the fact that we have defined Hausdorff dimension only for subsets of $\uint$ -- the result generalizes directly to \emph{compact} subsets of a metric space though.

\begin{theorem}
$\dim_\mathcal{H} \equivW \lim \star \lim$.
\end{theorem}
\begin{proof}
First, we show $\dim_\mathcal{H} \leqW \lim \star \lim$, split over the following two lemmata.

\begin{lemma}
$(d, A) \mapsto \{\sum_{i \in \mathbb{N}} r_i^d \mid \exists (x_i)_{i \in \mathbb{N}} \ A \subseteq \bigcup_{i \in \mathbb{N}} B(x_i,r_i) \ \} : \uint \times \mathcal{A}(\uint) \to \mathcal{O}(\mathbb{R}_<)$ is computable.
\begin{proof}
Given $d \in \mathbb{R}$ and $A \in \mathcal{K}(\uint)$ (using that $\uint$ is compact), we can compute $\{((r_0, r_1, \ldots), (x_0, x_1, \ldots)) \mid A \subseteq \bigcup_{i \in \mathbb{N}} B(x_i,r_i)\} \in \mathcal{O}(\uint^\mathbb{N} \times \uint^\mathbb{N})$. As $\uint^\mathbb{N}$ is computably overt, we can move to $\{(r_0, r_1, \ldots) \mid \exists (x_0, x_1, \ldots) A \subseteq \bigcup_{i \in \mathbb{N}} B(x_i,r_i)\} \in \mathcal{O}(\uint^\mathbb{N})$. For $d > 0$, pointwise exponentiation is an open map, thus we compute $\{(r_0^d, r_1^d, \ldots) \mid \exists (x_0, x_1, \ldots) A \subseteq \bigcup_{i \in \mathbb{N}} B(x_i,r_i)\} \in \mathcal{O}(\uint^\mathbb{N})$. Then we use Lemma \ref{lemma:opensum} to obtain $\{\sum_{i \in \mathbb{N}} r_i^d \mid \exists (x_i)_{i \in \mathbb{N}} \ A \subseteq \bigcup_{i \in \mathbb{N}} B(x_i,r_i) \ \} \in \mathcal{O}(\mathbb{R}_<)$.
\end{proof}
\end{lemma}

\begin{lemma}
$\left ( U_{(\cdot)} \mapsto \inf \{d \in \uint \mid \inf U_d = 0\} : \mathcal{C}(\uint, \mathcal{O}(\mathbb{R}_<)) \to \mathbb{R} \right ) \leqW \lim \star \lim$
\begin{proof}
We find that $\inf U_d = 0$ iff $\forall k \in \mathbb{N} \ 2^{-k} \in U_d$. We can compute $U_d \mapsto \{k \in \mathbb{N} \mid 2^{-k} \in U_d\} : \mathcal{O}(\mathbb{R}_<) \to \mathcal{O}(\mathbb{N})$. Now $\left ( (=\mathbb{N}?) : \mathcal{O}(\mathbb{N}) \to \{0,1\} \right ) \leqW \lpo \star \lim$, as follows from arguments in \cite{paulymaster,mylatz}. Since $\widehat{\lpo \star \lim} \equivW \lim \star \lim$, we can decide whether $\inf U_d = 0$ for all rational $d$ in parallel, and this suffices to obtain $\inf \{d \in \uint \mid \inf U_d = 0\} \in \mathbb{R}$.
\end{proof}
\end{lemma}

\begin{corollary}
$\left (A \mapsto \inf \{d \geq 0 \mid \inf \{\sum_{i \in \mathbb{N}} r_i^d \mid \exists (x_i)_{i \in \mathbb{N}} \ A \subseteq \bigcup_{i \in \mathbb{N}} B(x_i,r_i) \ \} = 0\} : \mathcal{A}(\uint) \to \mathbb{R} \right ) \leq_W \lim \star \lim$.
\end{corollary}

For the other direction, first note that a standard argument establishes:
\begin{lemma}
$\left ((q_i)_{i \in \mathbb{N}} \mapsto \sup_{j \in \mathbb{N}} \inf_{k \in \mathbb{N}} q_{\langle j,k\rangle} :\subseteq \mathcal{C}(\mathbb{N}, \mathbb{Q}) \to \mathbb{R} \right) \equivW \lim \star \lim$.
\begin{proof}
For the $\leqW$-direction, note that $\inf \equivW \sup \equivW \lim \equivW \widehat{\lim}$, e.g.~\cite{stein,brattka2}.

For the $\geqW$-direction, we employ the function $F_2 : \Cantor \to \Cantor$ defined via $F_2(x)(n) = 0$ iff $\forall i \exists k \ \ x(\langle n, i, k\rangle) = 0$ and $F_2(x)(n) = 1$ otherwise. This function was studied in \cite{stein,brattka}\footnote{In the previous literature, the function $F_2$ was denoted by $C_2$. We avoid this notation in order to prevent confusion with choice for the discrete two element space. Instead, our notation is inspired by \cite{nobrega}.}, and in \cite{gherardi4}, it is shown that $F_2 \equivW \lim \star \lim$.

Given $x \in \Cantor$, let $q_{\langle \langle n, j\rangle, k\rangle} := \sum_{l \leq n} 2*3^{-l} x(\langle l, j, k\rangle)$. Then $\sup_j \inf_k q_{\langle \langle n, j\rangle, k\rangle} = \sum_{\{l \leq n \mid \neg \forall j \exists k x(\langle l, j, k\rangle) = 0\}} 2*3^{-l}$, and subsequently, $\sup_{\langle n,j\rangle} \inf_k q_{\langle \langle n,j\rangle, k\rangle} = \sum_{\{l \in \mathbb{N} \mid \neg \forall j \exists k x(\langle l, j, k\rangle) = 0\}} 2*3^{-l}$. From such a member of the usual ternary Cantor set, the corresponding point in $\Cantor$ can be computed. This point then is $F_2(x)$.
\end{proof}
\end{lemma}

Given some closed set $A \in \mathcal{A}(\uint)$ and some interval $[a, b]$, let $A \prec [a,b]$ be the rescaling of $A$ into $[a,b]$. Not only is this a computable operation, but even $(A_i)_{i \in \mathbb{N}} \mapsto \{0\} \cup \bigcup_{i \in \mathbb{N}} A_i \prec [2^{-2i-2},2^{-2i-1}] : \mathcal{C}(\mathbb{N},\mathcal{A}(\uint)) \to \mathcal{A}(\uint)$ is computable. Moreover, we have:
\begin{lemma}
$\sup_{i \in \mathbb{N}} \dim_\mathcal{H}(A_i) = \dim_\mathcal{H}(\{0\} \cup \bigcup_{i \in \mathbb{N}} A_i \prec [2^{-2i-2},2^{-2i-1}])$
\end{lemma}

Combining this with the construction of sets of given Hausdorff dimension in Subsection \ref{subsec:introhausdorff} (which takes care of the inner $\inf$ \emph{for free}) and the preceding lemma, we obtain the claim.
\end{proof}

\subsection{The Frostman Lemma}
We now finally direct our attention to the maps $\operatorname{Frost}$ and $\operatorname{StrictFrost}$ introduced in Definition \ref{def:frostmanmaps}. One can prove the Frostman lemma via the min-cut/max-flow theorem, and the construction of the flows involved (if done in the right way) yields Weihrauch reductions to $\operatorname{NonZeroFlow}$ and $\operatorname{MaxFlow}$ respectively.

Given a closed set $A \in \mathcal{A}(\uint)$ and some $s \in \uint$ we construct an infinite binary tree with capacities. As in Subsection \ref{subsec:flows}, we assume that the vertices of the tree are labeled with (closed) dyadic intervals in the canonic way. As $\uint$ is compact, we can semidecide if $A \cap I = \emptyset$ for a closed interval, and we run all these tests simultaneously. When it comes to assigning a weight to a vertex of depths $n$, we choose $2^{-sn}$ if no interval containing the label of the vertex has been identified as disjoint with $A$ yet. If we have proof that some superset (and thus the label itself) is disjoint from $A$, we assign the weight $0$. Any cut through this tree gives an upper bound for the $s$-dimensional Hausdorff content of $A$, and so the min-cut/max-flow theorem implies that there is a non-zero flow compatible with the tree. Any such flow then induces an $s$-Frostman measure using the construction from Subsection \ref{subsec:flows}. This proves:

\begin{lemma}
$\operatorname{Frost} \leqW \operatorname{NonZeroFlow}$ and $\operatorname{StrictFrost} \leqW \operatorname{MaxFlow}$.
\end{lemma}

The converse directions are obtained from the following in conjunction with preceding results:

\begin{lemma}
$\C_\mathbb{N} \times \textrm{H}_1\textrm{C}_\uint \leqW \operatorname{Frost}$.
\begin{proof}
As before, we use the rescaling operation\footnote{This was defined as follows: Given some closed set $A \in \mathcal{A}(\uint)$ and some interval $[a, b]$, let $A \prec [a,b]$ be the rescaling of $A$ into $[a,b]$. This is a computable map.} $A \prec [a,b]$. Given a non-empty closed set $A \in \mathcal{A}(\mathbb{N})$ and another closed subset $B \in \mathcal{A}(\uint)$ with $\dim_\mathcal{H}(B) = 1$, we may compute the set $\{0\} \cup \bigcup_{n \in A} B \prec [2^{-2i-2}, 2^{-2i-1}] \in \mathcal{A}(\uint)$, and note that this set again has Hausdorff dimension $1$. Moreover, we make use of $\C_\mathbb{N} \equivW \textrm{UC}_\mathbb{N}$, i.e.~for closed choice on $\mathbb{N}$ we may safely assume that the closed set is a singleton. In this case, the support of any Frostman-measure on the set is inside $B \prec [2^{-2n-2},2^{-2n-1}]$ for the unique $n \in \mathbb{N}$ with $A = \{n\}$. From this, we can compute a point in $B \prec [2^{-2n-2},2^{-2n-1}]$, which in turn allows us to find both a point in $B$ as well as identify $n$.
\end{proof}
\end{lemma}

\begin{lemma}
$\left (\id : \mathcal{A}(\mathbb{N}) \to \mathcal{V}(\mathbb{N}) \right ) \leqW \operatorname{StrictFrost}$.
\begin{proof}
Given $A \in \mathcal{A}(\mathbb{N})$, we can compute $\left (\{0\} \cup \bigcup_{n \in A} [2^{-2n-2}, 2^{-2n-1}] \right )$. We use $\operatorname{StrictFrost}$ to find a measure with this set as support, and then Theorem \ref{theo:suppovert} to obtain the set as overt set. To complete the reduction, note $n \in A \Leftrightarrow (2^{-2n-2},2^{-2n-1}) \cap \left (\{0\} \cup \bigcup_{i \in A} [2^{-2i-2}, 2^{-2i-1}] \right ) \neq \emptyset$.
\end{proof}
\end{lemma}

\begin{corollary}
\label{corr:frostclassification}
$\C_\mathbb{N} \times \C_\uint \equivW \operatorname{Frost}$.
\end{corollary}

\begin{corollary}
$\lim \equiv_W \operatorname{StrictFrost}$.
\end{corollary}

A direct consequence of Corollary \ref{corr:frostclassification} is a computable closed set with Hausdorff dimension $> 2^{-n}$ may still fail to admit a computable $2^{-n}$-Frostman measure. Proposition \ref{prop:perfectcore} then shows that requiring the set to be a computable closed and overt set does not change this. On the other hand, in \cite{fouche} a construction of computable Frostman measures on very special sets is given. We can make the relevant properties of the sets explicit in the following:

\begin{lemma}
Given $s \in \uint$, a set $A \in \mathcal{A}(\uint) \wedge \mathcal{V}(\uint)$ and $p \in \Baire$ such that for any dyadic interval $I$ we find $I \cap A = \emptyset$ or $A \cap I$ admits an $s$-Frostman measure $\mu$ with $\mu(I \cap A) \geq 2^{-p(- \log |I|)}$ we can compute an $s$-Frostman measure on $A$.
\end{lemma}

By Propositions \ref{prop:concentrate}, \ref{prop:concentratedsupport} the converse is true, too:

\begin{corollary}
Let $A$ admit a computable $s$-Frostman measure. Then there is a computable $B \in \mathcal{A}(\uint) \wedge \mathcal{V}(\uint)$ with $B \subseteq A$ and a computable sequence $p \in \Baire$ such that for any dyadic interval $I$ we find $I \cap A = \emptyset$ or $A \cap I$ admits an $s$-Frostman measure $\mu$ with $\mu(I \cap A) \geq 2^{-p(- \log |I|)}$. Moreover, if an $s$-Frostman measure $\nu$ on $A$ is given, we can effectively find $B$ and $p$.
\end{corollary}

\subsection{On computably universally measure 0 sets}
Recall that a set $A$ is called \emph{universally measure 0}, if there is no non-atomic non-zero Radon measure supported by $A$. Clearly every countable set is universally measure $0$. There are universally measure 0 sets with cardinality $2^{\aleph_0}$, however, such sets cannot be Borel. An example on the real line was constructed by \name{Sierpi\'nski} and \name{Szpilrajn-Marczeswski} \cite{sierpinski}. Later, \name{Zindulka} even found an example of a universally measure 0 set with positive Hausdorff dimension \cite{zindulka} -- thus, the requirement of compactness of $A$ in the Frostman lemma cannot be completely relaxed.

Our results show that in the computable world, the picture is very different: If we call a set $A$ \emph{computably universally measure 0}, if there is no non-atomic non-zero computable Radon measure supported by $A$, we find that there is a computable closed and computable overt set with Hausdorff dimension $1$ that is computably universally measure $0$.

A much stronger effective notion is \emph{arithmetically universally measure 0} -- a set is this, if it does not support any non-atomic non-zero Radon measure computable relative to some arithmetic degree. This notion (under a different moniker) was employed by \name{Gregoriades} in \cite{gregoriades3} in order to study computable Polish spaces up to $\Delta_1^1$-isomorphisms: A Polish space $\mathbf{X}$ is $\Delta_1^1$-isomorphic to $\Baire$ iff it is not arithmetically universally measure $0$. \name{Gregoriades} then constructed computable closed and computable overt sets with cardinality $2^{\aleph_0}$ that are arithmetically universally measure 0.

\hide{
\subsection{More general constructions of measures}
Define the operator $\operatorname{ConstructMeasure} : \subseteq \mathcal{A}(\uint) \mto \mathcal{M}(\uint)$ by $\mu \in \operatorname{ConstructMeasure}(A)$ iff $\emptyset \neq \operatorname{supp}(\mu) \subseteq A$. Define $\operatorname{ConstructMeasure-NA/HD}$ by requiring $\mu$ to be non-atomic, but simultaneously restricting $A$ to have some fixed lower bound on its Hausdorff dimension. Then we obtain $\operatorname{StrictConstructMeasure}$ and $\operatorname{StrictConstructMeasure-NA/HD}$ by requesting $\operatorname{supp}(\mu) = A$ instead of merely $\operatorname{supp}(\mu) \subseteq A$ in the respective definitions. Finally, we find $\operatorname{StrictConstructMeasure-NA/HD+}$ by extending to all sets of positive Hausdorff dimension (without known bounds).

\begin{proposition}
$\operatorname{ConstructMeasure} \equivW \C_{\Cantor}$
\end{proposition}

\begin{proposition}
$\operatorname{ConstructMeasure-NA/HD} \equivW \C_{\Cantor} \times \C_\mathbb{N}$
\end{proposition}

\begin{proposition}
$\operatorname{StrictConstructMeasure} \equivW \lim \equivW \operatorname{StrictConstructMeasure-NA/HD+} \equivW \operatorname{StrictConstructMeasure-NA/HD}$
\end{proposition}}

\bibliographystyle{eptcs}
\bibliography{../spieltheorie}
\end{document}